\documentclass[11pt]{article}
\usepackage{graphicx}
\usepackage{amssymb,bm}
\usepackage{bbm}
\usepackage{amsmath}
\usepackage{amsfonts}
\usepackage[margin=1.0in]{geometry}
\usepackage{cite}
\usepackage{enumitem}

\DeclareMathOperator{\E}{\mathbbmss{E}}

\allowdisplaybreaks

\usepackage{amsthm}
\newtheorem{theorem}{Theorem}

\usepackage{chngcntr}
\counterwithout{theorem}{section}
\counterwithout{definition}{section}
\counterwithout{lemma}{section}
\counterwithout{remark}{section}
\counterwithout{assumption}{section}
\counterwithout{proposition}{section}
\counterwithout{corollary}{section}

\begin{document}
\title{On the Optimality of Treating Interference as Noise for Interfering Multiple Access Channels
\footnotetext{This work is partially supported by the U.K. Engineering and Physical Sciences Research Council (EPSRC) under grant EP/N015312/1.}}
\author{Hamdi~Joudeh and Bruno~Clerckx \\
\small
Communications and Signal Processing Group, Department of Electrical and Electronic Engineering \\
\small
Imperial College London, London SW7 2AZ, United Kingdom\\
\small
 Email: \{hamdi.joudeh10, b.clerckx\}@imperial.ac.uk}
\date{}
\maketitle
\begin{abstract}
In this paper, we look at the problem of treating interference as noise (TIN) in the Gaussian interfering multiple access channel (IMAC).
The considered network comprises $K$ mutually interfering multiple access channels (MACs), each consisting of two
transmitters communicating independent messages to one receiver.
We define the TIN scheme for this channel as one in which each MAC performs a power controlled version of its capacity-achieving strategy while treating interference
from all other MACs as noise.
We characterize an achievable generalized degrees-of-freedom (GDoF) region under
the TIN scheme and identify a regime of parameters (in terms of channel strength levels) where this region is optimal.
\end{abstract}
\section{Introduction}
\label{sec:introduction}
Transmitter power control coupled with treating interference as noise (TIN) at receivers
is one of the oldest and most commonly employed interference management strategies in wireless networks.
The TIN strategy derives its attractiveness from its (relatively) low complexity and its
robustness to channel uncertainty.
TIN was shown to achieve the sum-capacity of the $2$-user interference channel (IC)
in what is known as the noisy interference regime \cite{Annapureddy2009,Shang2009,Motahari2009}.
For the $K$-user IC, the problem is much more involved largely due to the intricate structure of the TIN-achievable
rate region \cite{Charafeddine2012} and the difficulty of the
underlying optimization problem \cite{Luo2008},
a surprising contrast to the simple structure of the TIN strategy itself.
This challenge was circumvented by Geng \emph{et al.} in \cite{Geng2015} through seeking an approximate solution based on the
generalized degrees-of-freedom (GDoF) \cite{Etkin2008}.

Geng \emph{et al.} identified a broad regime, described in terms of channel strength levels, where the TIN strategy achieves the exact GDoF region
and the entire capacity region  within a constant gap.
Beyond the regular $K$-user IC considered in \cite{Geng2015}, this type of TIN-optimality investigation, through the GDoF and capacity approximations,
has been extended in several directions
\cite{Geng2015a,Geng2016,Sun2016,Yi2016}.
Nevertheless, the optimality of TIN in cellular-like networks is an intriguing direction that remains meagerly investigated.
A recent result in this direction was reported in \cite{Gherekhloo2016}, where an alteration of the 2-user IC, termed the PIMAC, was considered.
The PIMAC consists of a point-to-point link and a 2-user multiple access channel (MAC) that interfere with each other.
The authors in \cite{Gherekhloo2016} identify regimes in which a simple time-sharing-TIN scheme is sum-GDoF optimal and achieves the
sum-capacity within a constant gap.
However, the specificity of the results and analysis in \cite{Gherekhloo2016} makes them difficult to generalize to settings
with more transmitters and receivers.

In this work, we consider an interfering multiple access channel (IMAC)
comprising $K$ mutually interfering 2-user MACs, e.g. Fig. \ref{fig:IMAC_example}.
This is a typical model for cellular networks operating in the uplink mode and subsumes
the setting in \cite{Gherekhloo2016}\footnote{Note that we do not claim that our
results subsume the ones in \cite{Gherekhloo2016}. In contrast to the entire GDoF region considered here,
the investigation in  \cite{Gherekhloo2016} is restricted to the sum-GDoF. This bears the possibility of
arriving at an enlarged TIN-optimality regime.}.
We introduce a TIN scheme in which each MAC performs a power controlled
version of its capacity-achieving strategy, while treating interference from all other MACs as noise.
We characterize an achievable GDoF region under the proposed TIN scheme.
Moreover, we identify a regime of channel parameters for which this region is optimal.
Finally, for the identified TIN-optimal regime, we show that the propose TIN scheme
achieves the entire capacity region to within a
constant gap.

\emph{Notation:}
For any positive integers $z_{1}$ and $z_{2}$ where $z_{1} \leq z_{2}$, the sets $\{1,2,\ldots,z_{1}\}$ and $\{z_{1},z_{1}+1,\ldots,z_{2}\}$ are denoted by
$\langle z_{1} \rangle$ and $\langle z_{1}:z_{2}\rangle$, respectively.
For any $a \in \mathbb{R}$, $(a)^{+} = \max\{0,a\}$.
Bold lowercase symbols  denote tuples, e.g. $\mathbf{a} = (a_{1},\ldots,a_{Z})$.
For $\mathcal{A} = \{\mathbf{a}_{1},\ldots,\mathbf{a}_{K}\}$, $\Sigma(\mathcal{A})$
is the set of all cyclicly ordered sequences of all subsets of $\mathcal{A}$,
e.g.
\begin{equation}
\nonumber
\Sigma\big(\{\mathbf{a}_{1},\mathbf{a}_{2},\mathbf{a}_{3}\}\big) =
\big\{(\mathbf{a}_{1}),(\mathbf{a}_{2}),(\mathbf{a}_{3}),(\mathbf{a}_{1},\mathbf{a}_{2}),(\mathbf{a}_{1},\mathbf{a}_{3}),(\mathbf{a}_{2},\mathbf{a}_{3}),
(\mathbf{a}_{1},\mathbf{a}_{2},\mathbf{a}_{3}),(\mathbf{a}_{1},\mathbf{a}_{3},\mathbf{a}_{2})  \big\}.
\end{equation}
\section{System Model and Preliminaries}
\label{sec:system model}
\begin{figure}
\centering
\includegraphics[width = 0.4\textwidth,trim={9cm 6cm 9cm 6cm},clip]{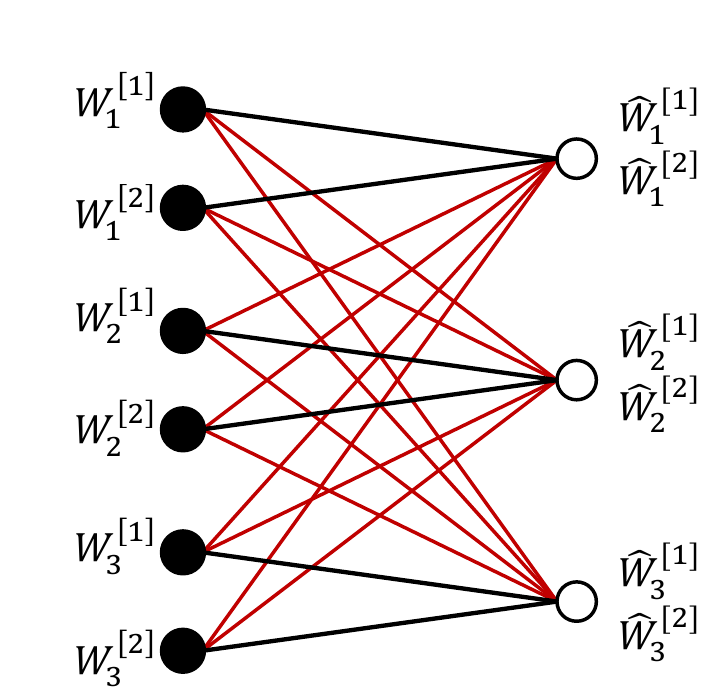}
\caption{A 3-cell interfering multiple access channel.}
\label{fig:IMAC_example}
\end{figure}
Consider a $K$-receiver Gaussian  IMAC in which each receiver is associated with $2$ transmitters.
The $k$-th receiver is denoted by Rx-$k$ and the $l_{k}$-th transmitter, $l_{k} \in \{1,2\}$, associated with this receiver is denoted by
Tx-$(l_{k},k)$.
We often use the terminology of cellular networks where a receiver and its associated transmitters are referred to as a cell.
The set of tuples corresponding to all transmitters (or users) in the network is given by
$\mathcal{K} \triangleq \left\{(l_{k},k) : l_{k} \in \{1,2\}, k \in \langle K \rangle \right\}$.

The input-output relationship at the $t$-th use of the channel is described as
\begin{equation}
\label{eq:system model}
Y_{i}(t) = \sum_{k = 1}^{K}  \Bigl[ h_{ki}^{[1]} \tilde{X}_{k}^{[1]}(t) + h_{ki}^{[2]} \tilde{X}_{k}^{[2]}(t) \Bigr] + Z_{i}(t),
\forall i \in \langle K \rangle
\end{equation}
where $h_{ki}^{[l_{k}]}$ is the channel coefficient from Tx-$(l_{k},k)$ to Rx-$i$,
$\tilde{X}_{k}^{[l_{k}]}(t)$ is the transmitted symbol of Tx-$(l_{k},k)$ and  $Z_{i}(t) \sim \mathcal{N}_{\mathbb{C}}(0,1)$ is the normalized additive white Gaussian noise (AWGN) at Rx-$i$.
All symbols are complex and each transmitter $(l_{k},k)$ is subject to the power constraint
$\E \big[|\tilde{X}_{k}^{[l_{k}]}(t)|^{2}\big] \leq P_{k}^{[l_{k}]}$.
Note that receivers are indexed by subscripts, transmitters are indexed by superscripts in square parentheses and channel
uses are indexed by arguments in round parentheses.

Following the standard reformulation in \cite{Geng2015}, the channel model in \eqref{eq:system model} is translated into
\begin{equation}
\label{eq:system model 2}
Y_{i}(t) = \sum_{k = 1}^{K} \Bigl[ \sqrt{ P^{\alpha_{ki}^{[1]}} } e^{j \theta_{ki}^{[1]}} X_{k}^{[1]}(t)
+\sqrt{ P^{\alpha_{ki}^{[2]}} } e^{j \theta_{ki}^{[2]}} X_{k}^{[2]}(t) \Bigr]
+ Z_{i}(t)
\end{equation}
where $P>0$ is a nominal power value and $X_{k}^{[l_{k}]}(t) = {\tilde{X}_{k}^{[l_{k}]}(t)}/{\sqrt{P_{k}^{[l_{k}]}}}$ is the normalized transmit symbol of Tx-$(l_{k},k)$ with power constraint $\E \big[|X_{k}^{[l_{k}]}(t) |^{2}\big] \leq 1$.
$\sqrt{ P^{\alpha_{ki}^{[l_{k}]}} }$ and $\theta_{ki}^{[l_{k}]}$ are the magnitude and phase of the channel between
 Tx-$(l_{k},k)$  and Rx-$i$, respectively. The exponent $\alpha_{ki}^{[l_{k}]}$ is known as the channel strength level, and is given by
\begin{equation}
\nonumber
\alpha_{ki}^{[l_{k}]} \triangleq \frac{\log \left( \max \big\{ 1, | h_{ki}^{[l_{k}]} |^{2} P_{k}^{[l_{k}]} \big\} \right) }{\log P}
, \ \forall (l_{k},k) \in \mathcal{K}, \
i\in \langle K \rangle.
\end{equation}
As in \cite{Geng2015}, avoiding negative channel strength levels has no impact on the results.
Without loss of generality, we assume the following order of direct link strength levels
\begin{equation}
\label{eq:strength_order}
\alpha_{kk}^{[1]} \leq  \alpha_{kk}^{[2]}, \ \forall k \in \langle K \rangle.
\end{equation}
\subsection{Messages, Rates, Capacity and GDoF}
Tx-$(1,k)$ and Tx-$(2,k)$ have the messages $W_{k}^{[1]}$ and $W_{k}^{[2]}$, respectively,
intended to Rx-$k$.
All messages are independent and $|W_{k}^{[l_{k}]}|$ denotes the size of the corresponding message set.
For codewords spanning $n$ channel uses, the rates $R_{k}^{[l_{k}]} = \frac{\log |W_{k}^{[l_{k}]}|}{n}$, $\forall (l_{k},k) \in \mathcal{K}$,
are achievable if  all messages can be decoded simultaneously with arbitrarily small error probability as $n$ grows sufficiently large.
A rate tuple is denoted by
$\mathbf{R} = \big(R_{1}^{[1]},R_{1}^{[2]},\ldots,R_{K}^{[1]},R_{K}^{[2]}\big)$
and the channel capacity region $\mathcal{C}$ is the closure of the set of all achievable rate tuples.
A GDoF tuple is denoted by $\mathbf{d} = \big(d_{1}^{[1]},d_{1}^{[2]},\ldots,d_{K}^{[1]},d_{K}^{[2]}\big)$
and the GDoF region is defined as
\begin{equation}
\nonumber
\mathcal{D} \triangleq \left\{ \mathbf{d} : \  d_{k}^{[l_{k}]} = \lim_{P \rightarrow \infty} \frac{R_{k}^{[l_{k}]}}{\log P}, \
\forall (l_{k},k) \in \mathcal{K}, \ \mathbf{R} \in \mathcal{C}  \right\}.
\end{equation}
\subsection{Treating (Inter-cell) Interference as Noise}
\label{subsec:TIN_scheme}
In the TIN scheme, a MAC-type capacity-achieving strategy is employed in each cell, with successive decoding of in-cell signals while treating all inter-cell interference as noise.
However, one key difference compared to the MAC (i.e. single-cell transmission) is that power control is employed by transmitters to manage inter-cell interference.
It is known that such power control is not required to achieve the corner points of the MAC capacity reagin \cite{Cover2012}.

Each transmitter Tx-$(l_{k},k)$ uses an independent Gaussian codebook and transmits with power $P^{r_{k}^{[l_{k}]}}$,
where $r_{k}^{[l_{k}]} \leq 0$ is the transmit power exponent.
On the other hand, each receiver Rx-$k$ performs successive decoding of its two desired signals, while treating all inter-cell interference as noise.
For a decoding order $\pi_{k}:\{1,2\} \rightarrow \{1,2\}$,
Rx-$k$ starts by decoding, and cancelling, $X_{k}^{[\pi_{k}(2)]}$ before decoding $X_{k}^{[\pi_{k}(1)]}$.
Hence, Tx-$\big(\pi_{k}(1),k\big)$ and Tx-$\big(\pi_{k}(2),k\big)$ achieve any rates $R_{k}^{[\pi_{k}(1)]}$ and $R_{k}^{[\pi_{k}(1)]}$, respectively,  that satisfy
\begin{align}
\label{eq:rate per user_1}
R_{k}^{[\pi_{k}(1)]} & \leq \log \Biggl( 1+ \frac{ P^{ r_{k}^{[\pi_{k}(1)]} +  \alpha_{kk}^{[\pi_{k}(1)]} }  }
{1  +
\sum_{j\neq k} \big[ P^{ r_{j}^{[1]} + \alpha_{jk}^{[1]} }+ P^{ r_{j}^{[2]} + \alpha_{jk}^{[2]} }  \big]  } \Biggr)  \\
\label{eq:rate per user_2}
R_{k}^{[\pi_{k}(2)]} & \leq \log \Biggl( 1+ \frac{ P^{ r_{k}^{[\pi_{k}(2)]} + \alpha_{kk}^{[\pi_{k}(2)]}  }  }
{1  + P^{ r_{k}^{[\pi_{k}(1)]} + \alpha_{kk}^{[\pi_{k}(1)]} }  +
\sum_{j\neq k} \big[ P^{ r_{j}^{[1]} + \alpha_{jk}^{[1]} }+ P^{ r_{j}^{[2]} + \alpha_{jk}^{[2]} }  \big]  } \Biggr).
\end{align}
In the GDoF sense, we have
\begin{align}
\label{eq:GDoF per user_1}
d_{k}^{[\pi_{k}(1)]}  &\leq \max \biggl\{ 0, r_{k}^{[\pi_{k}(1)]} + \alpha_{kk}^{[\pi_{k}(1)]}
 -  \Bigl(\max_{j \neq k} \bigl\{  \max_{l_{j}} \{ r_{j}^{[l_{j}]} + \alpha_{jk}^{[l_{j}]} \}  \bigr\} \Bigr)^{+}  \biggr\} \\
\label{eq:GDoF per user_2}
d_{k}^{[\pi_{k}(2)]}  &\leq \max \biggl\{  0, r_{k}^{[\pi_{k}(2)]}  +   \alpha_{kk}^{[\pi_{k}(2)]}
  -  \Bigl(\max \Bigl\{r_{k}^{[\pi_{k}(1)]}  +  \alpha_{kk}^{[\pi_{k}(1)]}  ,\max_{j \neq k} \bigl\{
  \max_{l_{j}} \{ r_{j}^{[l_{j}]}  +  \alpha_{jk}^{[l_{j}]} \}  \bigr\} \Bigr\} \Bigr)^{+} \biggr\}.
\end{align}

The decoding order across the network is defined as
$\bm{\pi} \triangleq \left( \pi_{1},\ldots,  \pi_{K} \right)$.
For a given $\bm{\pi}$, the \emph{TIN-achievable GDoF region}, denoted by $\mathcal{P}_{\bm{\pi}}^{\star}$,
is the set of all GDoF tuples $\mathbf{d}$ for which there exists a feasible
transmit power exponent tuple $\mathbf{r} \triangleq \big(r_{1}^{[1]},r_{1}^{[2]},\ldots,r_{K}^{[1]},r_{K}^{[2]}\big)$
such that \eqref{eq:GDoF per user_1} and \eqref{eq:GDoF per user_2} are satisfied for all
$k \in \langle K \rangle$.
The \emph{general TIN-achievable GDoF region} is defined as $\mathcal{P}^{\star} \triangleq \bigcup_{\bm{\pi}} \mathcal{P}_{\bm{\pi}}^{\star}$.
Note that any GDoF tuple in $\mathcal{P}^{\star}$ is achieved through a strategy identified by a decoding order and a power allocation, i.e. $(\bm{\pi},\mathbf{r})$, where no time-sharing between different strategies is allowed.

Similar to \cite{Geng2015}, we introduce a \emph{polyhedral TIN scheme}.
For a given $\bm{\pi}$, the corresponding \emph{polyhedral TIN-achievable GDoF region} $\mathcal{P}_{\bm{\pi}}$
is
described by all GDoF tuples that satisfy
\begin{align}
\label{eq:polyhedral_region_1}
r_{k}^{[\pi_{k}(l_{k})]} & \leq 0,  \  \forall (l_{k},k) \in \mathcal{K}  \\
\label{eq:polyhedral_region_2}
d_{k}^{[\pi_{k}(l_{k})]} & \geq 0,  \  \forall (l_{k},k) \in \mathcal{K}  \\
\label{eq:polyhedral_region_3}
d_{k}^{[\pi_{k}(1)]}  &\leq r_{k}^{[\pi_{k}(1)]} + \alpha_{kk}^{[\pi_{k}(1)]}
 -  \Bigl(\max_{j \neq k} \bigl\{  \max_{l_{j}} \{ r_{j}^{[l_{j}]} + \alpha_{jk}^{[l_{j}]} \}  \bigr\} \Bigr)^{+},
 \forall k \in \langle K \rangle \\
  \label{eq:polyhedral_region_4}
d_{k}^{[\pi_{k}(2)]}  &\leq r_{k}^{[\pi_{k}(2)]}  +   \alpha_{kk}^{[\pi_{k}(2)]}
  -   \Bigl(\max \Bigl\{r_{k}^{[\pi_{k}(1)]}  +  \alpha_{kk}^{[\pi_{k}(1)]}  ,\max_{j \neq k} \bigl\{
  \max_{l_{j}} \{ r_{j}^{[l_{j}]}  +  \alpha_{jk}^{[l_{j}]} \}  \bigr\} \Bigr\} \Bigr)^{+},
  \forall k \in \langle K \rangle
\end{align}
where the first $\max\{0,\cdot\}$ in \eqref{eq:GDoF per user_1} and \eqref{eq:GDoF per user_2}  has been dropped.
It follows from this restriction that $\mathcal{P}_{\bm{\pi}} \subseteq \mathcal{P}_{\bm{\pi}}^{\star}$.
Taking the union over all possible decoded orders, we achieve the region given by $\mathcal{P} = \bigcup_{\bm{\pi}} \mathcal{P}_{\bm{\pi}}$.
It is readily seen that $\mathcal{P} \subseteq \mathcal{P}^{\star} \subseteq \mathcal{D}$.

As it turns out, for any $\bm{\pi}$, the region $\mathcal{P}_{\bm{\pi}}$ is a polyhedron (see Theorem \ref{theorem:TIN_region} in the following section). However, $\mathcal{P}$ is not a polyhedron in general, since it is a union of multiple polyhedra.
Yet, every GDoF point in $\mathcal{P}$ is achieved by fixing
$\bm{\pi}$ and applying a polyhedral TIN scheme with power allocation $\mathbf{r}$ satisfying \eqref{eq:polyhedral_region_1}-\eqref{eq:polyhedral_region_4}.

In the following, we often work with the identity order $\bm{\pi} = \bm{\mathrm{id}}$, where
$\bm{\mathrm{id}} \triangleq \left( \mathrm{id}_{1},\ldots,\mathrm{id}_{K} \right)$ and $\mathrm{id}_{i}(l_{i}) = l_{i}$,
$\forall(l_{i},i) \in \mathcal{K}$.
The corresponding polyhedral TIN region is denoted by $\mathcal{P}_{\bm{\mathrm{id}}}$.
\subsection{Some Known Special Cases}
\label{subsec:special_cases}
Before presenting the main results, we review some known GDoF region characterizations for subnetworks of the considered IMAC.
First, we consider a regular IC obtained by removing one transmitter from
each cell and leaving only Tx-$(l_{i}, i)$, $i \in \langle K \rangle$.
From \cite{Geng2015}, the polyhedral TIN region for this subnetwork is given by
\begin{align}
\label{eq:IC_GDoF_region_1}
0 \leq  d_{i}^{[l_{i}]} & \leq \alpha_{ii}^{[l_{i}]}, \ \forall i \in \langle K \rangle\\
\label{eq:IC_GDoF_region_2}
\sum_{j \in \langle m \rangle} d_{i_{j}}^{[l_{i_{j}}]} & \leq
\sum_{j \in \langle m \rangle} \alpha_{i_{j}i_{j}}^{[l_{i_{j}}]} - \alpha_{i_{j}i_{j-1}}^{[l_{i_{j}}]},
\
(i_{1},\ldots,i_{m}) \in \Sigma \big(\langle K \rangle \big), m \in \langle 2:K \rangle
\end{align}
where the set of cyclic sequences $\Sigma \big(\langle K \rangle \big)$ is defined in the notation part of Section \ref{sec:introduction}
and a modulo operation is implicitly used on receiver indices when dealing with cyclic sequences, e.g. $i_{0} = i_{m}$.
The region in \eqref{eq:IC_GDoF_region_1}--\eqref{eq:IC_GDoF_region_2} is optimal for the regular IC under the TIN-optimality conditions in \cite{Geng2015}.

Next, consider the MAC consisting of Rx-$i$ and its
transmitters Tx-$(1,i)$ and Tx-$(2,i)$. The GDoF region achieved while fixing the decoding order $\pi_{i}$ is given by
\begin{align}
\label{eq:MAC_GDoF_region_1}
d_{i}^{[l_{i}]} & \geq 0, \ \forall l_{i} \in \{1,2\}\\
\label{eq:MAC_GDoF_region_2}
\sum_{s_{i} \in \langle l_{i} \rangle} d_{i}^{[\pi_{i}(s_{i})]} & \leq \alpha_{ii}^{[\pi_{i}(l_{i})]}, \ \forall l_{i} \in \{1,2\}.
\end{align}
It can be easily checked that the optimal GDoF region of the considered MAC is given by \eqref{eq:MAC_GDoF_region_1}--\eqref{eq:MAC_GDoF_region_2}
while fixing the decoding order to $\pi_{i} = \mathrm{id}$.
The signal of the \emph{stronger} user, i.e. Tx-$(2,i)$, is always received at a higher power level and is hence decoded first, making the other order
redundant from a GDoF perspective. Note that
this is in contrast to the MAC capacity region, which requires changing
the successive decoding order to achieve different corner points in general \cite{Cover2012}.
\section{Main Results}
\label{sec:main_results}
Here we present the main results of the paper.
\begin{theorem}
\label{theorem:TIN_region}
For the IMAC described in Section \ref{sec:system model},
the polyhedral TIN-achievable GDoF region $\mathcal{P}_{\bm{\pi}}$, for any decoding order $\bm{\pi}$,
is given by all tuples $\mathbf{d}$ that satisfy
\begin{align}
\label{eq:IMAC_TIN_region_1}
d_{i}^{[l_{i}]}   & \geq 0, \ \forall (l_{i},i) \in \mathcal{K} \\
\label{eq:IMAC_TIN_region_2}
\sum_{s_{i} \in \langle l_{i} \rangle} d_{i}^{[\pi_{i}(s_{i})]} & \leq \alpha_{ii}^{[\pi_{i}(l_{i})]}, \ \forall (l_{i},i) \in \mathcal{K} \\
\nonumber
\sum_{j \in \langle m \rangle }  \! \sum_{s_{i_{j}} \! \in \! \langle l_{i_{j}} \rangle}
\! \! \! d_{i_{j}}^{[\pi_{i_{j}}(s_{i_{j}})]} &  \! \leq \! \! \! \!
\sum_{j \in \langle m \rangle } \alpha_{i_{j}i_{j}}^{[\pi_{i_{j}}(l_{i_{j}})]} \! - \alpha_{i_{j}i_{j-1}}^{[\pi_{i_{j}}(l_{i_{j}})]}, \\
\label{eq:IMAC_TIN_region_3}
 \forall l_{i_{j}} \in \{1,2\}, (i_{1},\ldots,&i_{m})   \in \Sigma(\langle K \rangle), m \in \langle 2:K \rangle.
\end{align}
\end{theorem}
In  \eqref{eq:IMAC_TIN_region_3}, a modulo operation is implicitly used on receiver indices, e.g. $i_{0} = i_{m}$.
The proof of Theorem \ref{theorem:TIN_region} is presented in Section \ref{sec:TIN_region}.
It can be seen that the characterization of $\mathcal{P}_{\bm{\pi}}$ in Theorem \ref{theorem:TIN_region} inherits the features of
both the IC and MAC characterizations presented  in Section \ref{subsec:special_cases}.
Moreover, in contrast to the MAC, the decoding order $\bm{\mathrm{id}}$
does not necessarily yield the largest polyhedral region for the IMAC, i.e.
$\mathcal{P}_{\bm{\pi}} \subseteq \mathcal{P}_{\bm{\mathrm{id}}}$ does not hold in general for all  $\bm{\pi}$.
This inclusion, however, holds under the TIN-optimality conditions presented in the following result.
\begin{theorem}
\label{theorem:TIN_optimality}
For the IMAC described in Section \ref{sec:system model}, if the following conditions are satisfied
\begin{align}
\label{eq:TIN_condition_1}
\alpha_{ii}^{[l_{i}]} & \geq \max_{j:j \neq i} \left\{\alpha_{ij}^{[l_{i}]} \right\} +
\max_{(l_{k},k):k \neq i} \left\{\alpha_{ki}^{[l_{k}]} \right\}, \ \forall(l_{i},i),(l_{k},k) \in \mathcal{K}, \; j \in \langle K\rangle  \\
\label{eq:TIN_condition_2}
\alpha_{ii}^{[2]} - \alpha_{ij}^{[2]} & \geq \alpha_{ii}^{[1]} - \alpha_{ij}^{[1]} + \min\left\{\alpha_{ij}^{[1]},\alpha_{ij}^{[2]} \right\}, \
\forall i,j \in \langle K \rangle, \; i \neq j,
\end{align}
then the optimal GDoF region is given by $\mathcal{P}_{\bm{\mathrm{id}}}$, achieved through the polyhedral
TIN scheme in Section \ref{subsec:TIN_scheme}, and described by \eqref{eq:IMAC_TIN_region_1}--\eqref{eq:IMAC_TIN_region_3}
while setting $\bm{\pi} = \bm{\mathrm{id}}$.
\end{theorem}
The proof of Theorem \ref{theorem:TIN_optimality} is given in Section \ref{sec:TIN Optimality}.
The condition in \eqref{eq:TIN_condition_1} is essentially the one identified by Geng \emph{et al.} in \cite{Geng2015},
applied to all  regular IC subnetworks of the IMAC.
On the other hand, a special case of \eqref{eq:TIN_condition_2}
was identified by Gherekhloo \emph{et al.} in \cite{Gherekhloo2016} for the PIMAC described in Section \ref{sec:introduction}.
Note that under the above TIN conditions, we have $\mathcal{D} = \mathcal{P}^{\star} = \mathcal{P}  = \mathcal{P}_{\bm{\mathrm{id}}}$.

Before we proceed, it is worthwhile highlighting that as pointed out in \cite[Remark 1]{Geng2016}, existing TIN-optimality results are
\emph{``primarily in the form of sufficient conditions''} and that the necessity of such conditions
\emph{``remains undetermined in most cases''}.
The TIN-optimality result in Theorem \ref{theorem:TIN_optimality} is no exception to most existing results in that regards.
\section{TIN-Achievable GDoF Region}
\label{sec:TIN_region}
In this part, we prove Theorem \ref{theorem:TIN_region} by constructing a potential graph \cite{Geng2015,Geng2016} for the considered IMAC
and invoking the potential theorem \cite{Schrijver2002}.
To avoid cumbersome notation, we work with $\mathcal{P}_{\bm{\mathrm{id}}}$.
All derivations extend to $\mathcal{P}_{\bm{\pi}}$ by replacing each superscript $l_{k}$ with the corresponding
$\pi_{k}(l_{k})$.
\subsection{Feasible Power Allocation}
The first step towards applying the potential theorem is to derive the conditions of feasible power allocation.
To this end, we rewrite the inequalities in \eqref{eq:polyhedral_region_3} and \eqref{eq:polyhedral_region_4} as
\begin{align}
\label{eq:GDoF_polyhedral_1}
d_{k}^{[1]}  &\leq \min\biggl\{ r_{k}^{[1]} + \alpha_{kk}^{[1]} ,\min_{j \neq k} \Bigl\{ \min\bigl\{
r_{k}^{[1]} - r_{j}^{[1]} + \alpha_{kk}^{[1]} - \alpha_{jk}^{[1]},
r_{k}^{[1]} - r_{j}^{[2]} + \alpha_{kk}^{[1]} - \alpha_{jk}^{[2]}\bigr\}  \Bigr\} \biggr\} \\
\nonumber
d_{k}^{[2]}  &\leq \min\biggl\{ r_{k}^{[2]} + \alpha_{kk}^{[2]},
\min_{j \neq k} \Bigl\{ \min\bigl\{
r_{k}^{[2]} - r_{j}^{[1]} + \alpha_{kk}^{[2]} - \alpha_{jk}^{[1]},
r_{k}^{[2]} - r_{j}^{[2]} + \alpha_{kk}^{[2]} - \alpha_{jk}^{[2]}\bigr\}\Bigr\} , \\
\label{eq:GDoF_polyhedral_2}
& \quad \quad
r_{k}^{[2]} - r_{k}^{[1]} + \alpha_{kk}^{[2]} - \alpha_{kk}^{[1]}   \biggr\}.
\end{align}
From \eqref{eq:GDoF_polyhedral_1} and \eqref{eq:GDoF_polyhedral_2}, it follows that the polyhedral TIN region $\mathcal{P}_{\bm{\mathrm{id}}}$, described by the inequalities in \eqref{eq:polyhedral_region_1}--\eqref{eq:polyhedral_region_4} while setting $\bm{\pi} = \bm{\mathrm{id}}$, is equivalently described by the following inequalities
\begin{align}
\label{eq:polyhedral_region_21}
r_{k}^{[l_{k}]} & \leq 0,  \  \forall (l_{k},k) \in \mathcal{K}  \\
\label{eq:polyhedral_region_22}
d_{k}^{[l_{k}]} & \geq 0,  \  \forall (l_{k},k) \in \mathcal{K}  \\
\label{eq:polyhedral_region_23}
d_{k}^{[l_{k}]} & \leq \alpha_{kk}^{[l_{k}]} + r_{k}^{[l_{k}]},
\ \forall (l_{k},k) \in \mathcal{K}  \\
\label{eq:polyhedral_region_24}
d_{k}^{[l_{k}]} & \leq  r_{k}^{[l_{k}]}  -  r_{j}^{[l_{j}]} + \alpha_{kk}^{[l_{k}]} -  \alpha_{jk}^{[l_{j}]} , \
\forall (l_{k},k), (l_{j},j) \in \mathcal{K}, \ j \neq k \\
\label{eq:polyhedral_region_25}
d_{k}^{[2]} & \leq r_{k}^{[2]} -   r_{k}^{[1]} + \alpha_{kk}^{[2]} - \alpha_{kk}^{[1]}, \ \forall k \in \langle K \rangle.
\end{align}
After rearranging, the inequalities in \eqref{eq:polyhedral_region_21}--\eqref{eq:polyhedral_region_25} are rewritten as
\begin{align}
\label{eq:polyhedral_region_31}
d_{k}^{[l_{k}]} & \geq 0,  \  \forall (l_{k},k) \in \mathcal{K}  \\
\label{eq:polyhedral_region_32}
r_{k}^{[l_{k}]} & \leq 0,  \  \forall (l_{k},k) \in \mathcal{K}  \\
\label{eq:polyhedral_region_33}
-r_{k}^{[l_{k}]} & \leq \alpha_{kk}^{[l_{k}]} - d_{k}^{[l_{k}]},
\ \forall (l_{k},k) \in \mathcal{K}  \\
\label{eq:polyhedral_region_34}
r_{j}^{[l_{j}]} - r_{k}^{[l_{k}]}  & \leq \alpha_{kk}^{[l_{k}]} -  \alpha_{jk}^{[l_{j}]} - d_{k}^{[l_{k}]}, \
\forall (l_{k},k), (l_{j},j) \in \mathcal{K}, \ j \neq k \\
\label{eq:polyhedral_region_35}
r_{k}^{[1]}   - r_{k}^{[2]} & \leq \alpha_{kk}^{[2]} - \alpha_{kk}^{[1]} - d_{k}^{[2]}, \ \forall k \in \langle K \rangle.
\end{align}
Hence, a GDoF tuple $\mathbf{d} \in \mathbb{R}_{+}^{2K}$ is in the polyhedral TIN region
$\mathcal{P}_{\bm{\mathrm{id}}}$ if and only if there exists a power allocation tuple $\mathbf{r}\in \mathbb{R}^{2K}$ such that \eqref{eq:polyhedral_region_32}--\eqref{eq:polyhedral_region_35} hold.
\subsection{Potential Graph}
Next, we construct the potential graph \cite{Geng2015,Geng2016}.
This is a directed graph (digraph) $\mathcal{G}_{\mathrm{p}} = (\mathcal{V}, \mathcal{E} )$,
where the set of vertices $\mathcal{V}$ and the set of directed edges (or edges) $\mathcal{E}$
are given by
\begin{align}
\mathcal{V} & = \{u\} \cup \left\{ v_{k}^{[l_{k}]} : (l_{k},k) \in \mathcal{K} \right\} \\
\mathcal{E} & = \mathcal{E}_{1}' \cup \mathcal{E}_{1}'' \cup  \mathcal{E}_{2} \cup \mathcal{E}_{3}' \cup \mathcal{E}_{3}'' \\
\mathcal{E}_{1}' & = \left\{\big(v_{k}^{[1]} , v_{k}^{[2]} \big) :
k \in \langle K \rangle \right\} \\
\mathcal{E}_{1}'' & = \left\{\big(v_{k}^{[2]} , v_{k}^{[1]} \big) :
k \in \langle K \rangle  \right\} \\
\mathcal{E}_{2} & = \left\{\big(v_{k}^{[l_{k}]} , v_{j}^{[l_{j}]} \big) : (l_{k},k), (l_{j},j) \in \mathcal{K}, \ k \neq j \right\} \\
\mathcal{E}_{3}' & = \left\{ \big(u ,v_{k}^{[l_{k}]} \big) : (l_{k},k) \in \mathcal{K} \right\} \\
\mathcal{E}_{3}'' & = \left\{\big(v_{k}^{[l_{k}]} , u \big) : (l_{k},k) \in \mathcal{K} \right\}.
\end{align}
We define the length function $l:\mathcal{E} \rightarrow \mathbb{R}$ and assign the following lengths to different edges
\begin{align}
l\big(v_{k}^{[1]} , v_{k}^{[2]} \big) & = \alpha_{kk}^{[1]} - d_{k}^{[1]}, \ \forall k \in \langle K \rangle  \\
l\big(v_{k}^{[2]} , v_{k}^{[1]} \big) & = \alpha_{kk}^{[2]} - \alpha_{kk}^{[1]} -d_{k}^{[2]}, \ \forall k \in \langle K \rangle   \\
l\big(v_{k}^{[l_{k}]} , v_{j}^{[l_{j}]} \big) & =  \alpha_{kk}^{[l_{k}]} -  \alpha_{jk}^{[l_{j}]} - d_{k}^{[l_{k}]}, \
\forall (l_{k},k), (l_{j},j) \in \mathcal{K}, \ k \neq j \\
l\big(v_{k}^{[l_{k}]} , u \big) & =  \alpha_{kk}^{[l_{k}]}  - d_{k}^{[l_{k}]}, \
\forall (l_{k},k) \in \mathcal{K} \\
l\big( u, v_{k}^{[l_{k}]} \big) & =  0, \
\forall (l_{k},k) \in \mathcal{K}.
\end{align}
By definition \cite{Schrijver2002}, the function $p: \mathcal{V} \rightarrow \mathbb{R}$ is called a potential if for any pair of vertices
$a,b \in \mathcal{V}$ such that $(a,b) \in \mathcal{E}$, we have $l(a,b) \geq p(b) - p(a)$.
These conditions depend only on the difference between potential function values. Therefore, if there exists a valid potential function, we may assume without loss of generality that $p(u) = 0$, i.e. vertex $u$ is set as the ground.
By setting $p\big(v_{k}^{[l_{k}]}\big) = r_{k}^{[l_{k}]}$, it can be seen that such potential function values should satisfy
\begin{align}
\label{eq:potential_polyhedral_region_1}
r_{k}^{[2]}  - r_{k}^{[1]} & \leq
\alpha_{kk}^{[1]} - d_{k}^{[1]}, \ \forall k \in \langle K \rangle  \\
\label{eq:potential_polyhedral_region_2}
r_{k}^{[1]}   - r_{k}^{[2]} & \leq
\alpha_{kk}^{[2]} - \alpha_{kk}^{[1]} - d_{k}^{[2]}, \ \forall k \in \langle K \rangle \\
\label{eq:potential_polyhedral_region_3}
r_{j}^{[l_{j}]} - r_{k}^{[l_{k}]}  & \leq
\alpha_{kk}^{[l_{k}]} -  \alpha_{jk}^{[l_{j}]} - d_{k}^{[l_{k}]}, \
\forall (l_{k},k), (l_{j},j) \in \mathcal{K}, \ j \neq k \\
\label{eq:potential_polyhedral_region_4}
-r_{k}^{[l_{k}]} & \leq
\alpha_{kk}^{[l_{k}]} - d_{k}^{[l_{k}]}, \ \forall (l_{k},k) \in \mathcal{K} \\
\label{eq:potential_polyhedral_region_5}
r_{k}^{[l_{k}]} & \leq 0, \ \forall (l_{k},k) \in \mathcal{K}.
\end{align}
It is easily to check that the inequalities in \eqref{eq:potential_polyhedral_region_2}--\eqref{eq:potential_polyhedral_region_5}
are equivalent to the ones in \eqref{eq:polyhedral_region_32}--\eqref{eq:polyhedral_region_35}.
Moreover, the inequality in \eqref{eq:potential_polyhedral_region_1} is obtained by adding the inequalities in
\eqref{eq:polyhedral_region_32} and \eqref{eq:polyhedral_region_33}.
Therefore, it follows that  \emph{$\mathbf{d} \in \mathbb{R}_{+}^{2K}$ is in $\mathcal{P}_{\bm{\mathrm{id}}}$
if and only if there exists a valid potential function for $\mathcal{G}_{\mathrm{p}}$}.
At this point, we are ready to invoke the \emph{potential theorem \cite[Th. 8.2]{Schrijver2002}: there exists a potential function for a digraph $\mathcal{G}_{\mathrm{p}}$ if and only if each directed circuit in $\mathcal{G}_{\mathrm{p}}$ has a non-negative length.}

From the above, we conclude that the GDoF tuple $\mathbf{d} \in \mathbb{R}_{+}^{2K}$ is in
the polyhedral region $\mathcal{P}_{\bm{\mathrm{id}}}$ if and only if the length of each directed circuit in the potential graph $\mathcal{G}_{\mathrm{p}}$ is non-negative.
\subsection{Directed Circuits and GDoF Inequalities}
In this part, we examine all valid directed circuits (or circuits for short) of $\mathcal{G}_{\mathrm{p}}$ and derive the corresponding GDoF
inequalities.
When dealing with circuits, we refer to a vertex of type $v_{i}^{[l_{i}]}$ as a user.
It is readily seen that circuits of $\mathcal{G}_{\mathrm{p}}$ can be categorized into single-cell circuits and multi-cell circuits,
 depending on the participating users, as we see in what follows.
\subsubsection{Single-Cell Circuits}
Such circuits involve users belonging to only one cell and can be further categorized into:
\begin{itemize}
\item Single-user circuits of the form $\big( u \rightarrow v_{i}^{[l_{i}]} \rightarrow u \big)$, $\forall (l_{i},i) \in \mathcal{K}$.
 From the non-negative length condition, each of such circuits yields a single user bound given by
\begin{equation}
\label{eq:GDoF_single_circuit_1}
d_{i}^{[l_{i}]} \leq \alpha_{ii}^{[l_{i}]}.
\end{equation}
\item Multi-user circuits of the form $\big( u \rightarrow v_{i}^{[2]} \rightarrow  v_{i}^{[1]}  \rightarrow u \big)$ or
 $\big( v_{i}^{[2]}  \rightarrow  v_{i}^{[1]}  \rightarrow v_{i}^{[2]}  \big)$, $\forall i \in \langle K \rangle$.
 From the non-negative length condition applied to such circuits, we obtain
\begin{equation}
\label{eq:GDoF_single_circuit_2}
d_{i}^{[1]} + d_{i}^{[2]}  \leq \alpha_{ii}^{[2]}.
\end{equation}
\item Multi-user circuits of the form $\big( u \rightarrow v_{i}^{[1]}  \rightarrow  v_{i}^{[2]}  \rightarrow u \big)$, $\forall i \in \langle K \rangle$,
 from which we obtain
\begin{equation}
\label{eq:GDoF_single_circuit_3}
d_{i}^{[1]} + d_{i}^{[2]}  \leq \alpha_{ii}^{[1]} +  \alpha_{ii}^{[2]}.
\end{equation}
\end{itemize}
It can be seen that for $l_{i} = 2$, the GDoF inequality in \eqref{eq:GDoF_single_circuit_1} is
redundant since it is implied by \eqref{eq:GDoF_single_circuit_2}.
Moreover, the inequality in \eqref{eq:GDoF_single_circuit_3} is loose in general compared to the one in \eqref{eq:GDoF_single_circuit_2}.
\subsubsection{Multi-Cell Circuits}
Such circuits involve users belonging to more than one cell.
In particular, consider a cyclic sequence of tuples given by $\big( (l_{1},i_{1}), \ldots,(l_{n},i_{n}) \big) \in \Sigma(\mathcal{K})$, such that
 $i_{j'} \neq i_{j''}$ for some $j',j'' \in \langle n \rangle$.
The corresponding multi-cell circuit of $\mathcal{G}_{\mathrm{p}}$ takes one of the two following forms.
\begin{itemize}
\item Does not traverse $u$: $\big( v_{i_{0}}^{[l_{0}]} \rightarrow v_{i_{1}}^{[l_{1}]} \rightarrow \cdots  \rightarrow v_{i_{n}}^{[l_{n}]}  \big)$,
 where $(l_{0},i_{0}) = (l_{n},i_{n})$.
\item Traverses  $u$: $\big( u \rightarrow v_{i_{1}}^{[l_{1}]} \rightarrow \cdots  \rightarrow v_{i_{n}}^{[l_{n}]}  \rightarrow u  \big)$.
\end{itemize}
From the non-negative length condition, a multi-cell circuit of the first form yields the inequality
\begin{equation}
\label{eq:GDoF_circuit_form_1}
\sum_{j \in \langle n \rangle} d_{i_{j}}^{[l_{j}]}
\leq
\sum_{j \in \langle n \rangle}  \alpha_{i_{j}i_{j}}^{[l_{j}]} - \alpha_{i_{j+1}i_{j}}^{[l_{j+1}]}
\mathbbm{1}_{\mathcal{E}_{1}'^{\mathrm{c}}}\big(v_{i_{j}}^{[l_{j}]}, v_{i_{j+1}}^{[l_{j+1}]} \big)
\end{equation}
where
\begin{equation}
\nonumber
\mathbbm{1}_{\mathcal{E}_{1}'^{\mathrm{c}}}\big(v_{i}^{[l_{i}]}, v_{k}^{[l_{k}]} \big) =
\begin{cases}
0, \ \text{if}  \ \big(v_{i}^{[l_{i}]}, v_{k}^{[l_{k}]} \big)  \in \mathcal{E}_{1}' \\
1, \ \text{otherwise}.
\end{cases}
\end{equation}
Note that a modulo operation is used in \eqref{eq:GDoF_circuit_form_1}, and throughout this part, such that
$i_{n+1} = i_{1}$ and $l_{n+1} = l_{1}$.
On the other hand, a circuit of the second form gives the inequality
\begin{equation}
\label{eq:GDoF_circuit_form_2}
\sum_{j \in \langle n \rangle} d_{i_{j}}^{[l_{j}]}
\leq
\alpha_{i_{n}i_{n}}^{[l_{n}]} + \sum_{j =1}^{n-1}  \alpha_{i_{j}i_{j}}^{[l_{j}]} - \alpha_{i_{j+1}i_{j}}^{[l_{j+1}]}
\mathbbm{1}_{\mathcal{E}_{1}'^{\mathrm{c}}}\big(v_{i_{j}}^{[l_{j}]}, v_{i_{j+1}}^{[l_{j+1}]} \big)
\end{equation}
It is readily seen that \eqref{eq:GDoF_circuit_form_1} is tighter than \eqref{eq:GDoF_circuit_form_2} as
$\alpha_{i_{1}i_{n}}^{[l_{1}]}\mathbbm{1}_{\mathcal{E}_{1}'^{\mathrm{c}}}\big(v_{i_{n}}^{[l_{n}]}, v_{i_{1}}^{[l_{1}]} \big) \geq 0$.
Therefore, it is sufficient to consider multi-cell circuits that do not traverse $u$.

Next, we show that the GDoF inequality in \eqref{eq:GDoF_circuit_form_1} is necessarily redundant if the underlying circuit belongs to
at least one of the following classes:
\begin{enumerate}[label=C.{\arabic*}]
\item Circuits that traverse at least one edge in $\mathcal{E}_{1}'$, i.e. with two cyclicly adjacent users that belong to the same cell
$i \in \langle K \rangle$ and $v_{i}^{[1]}$ preceded $v_{i}^{[2]}$ in the cyclic order.
\label{cond:redundant_1}
\item Circuits that traverse $v_{i_{j}}^{[l_{j}]}$ and $v_{i_{k}}^{[l_{k}]}$, where $i_{j} = i_{k}$, $j \neq k+1$ and $k \neq j+1$,
i.e. with two cyclicly non-adjacent users that belong to the same cell.
\label{cond:redundant_2}
\item Circuits that traverse $v_{i}^{[2]}$, for some $i \in \langle K \rangle$, and do not traverse $v_{i}^{[1]}$.
\label{cond:redundant_3}
\end{enumerate}
First, suppose that we have a circuit in class \ref{cond:redundant_1}.
We may assume, without loss of generality, that $\big(v_{i_{1}}^{[l_{1}]}, v_{i_{2}}^{[l_{2}]} \big) \in \mathcal{E}_{1}'$, i.e. $i_{1} = i_{2}$, $l_{1} = 1$ and  $l_{2} = 2$.
The corresponding GDoF inequality is given by
\begin{equation}
\label{eq:GDoF_circuit_C1_1}
d_{i_{1}}^{[1]}  + d_{i_{1}}^{[2]} + \sum_{j=3}^{n} d_{i_{j}}^{[l_{j}]}
\leq
\alpha_{i_{1}i_{1}}^{[1]} + \alpha_{i_{1}i_{1}}^{[2]} - \alpha_{i_{3}i_{1}}^{[l_{3}]}
+
\sum_{j = 3}^{n}  \alpha_{i_{j}i_{j}}^{[l_{j}]} - \alpha_{i_{j+1}i_{j}}^{[l_{j+1}]}
\mathbbm{1}_{\mathcal{E}_{1}'^{\mathrm{c}}}\big(v_{i_{j}}^{[l_{j}]}, v_{i_{j+1}}^{[l_{j+1}]} \big).
\end{equation}
Now consider the circuits
$\big( v_{i_{0}}^{[l_{0}]} \rightarrow v_{i_{1}}^{[1]} \rightarrow v_{i_{3}}^{[l_{3}]}  \rightarrow  \cdots  \rightarrow v_{i_{n}}^{[l_{n}]}\big)$
and
$\big( u \rightarrow v_{i_{1}}^{[2]} \rightarrow u \big)$.
These are valid circuits of $\mathcal{G}_{\mathrm{p}}$, and give rise to the GDoF inequalities
\begin{align}
\label{eq:GDoF_circuit_C1_2}
d_{i_{1}}^{[1]}  + \sum_{j=3}^{n} d_{i_{j}}^{[l_{j}]}
& \leq
\alpha_{i_{1}i_{1}}^{[1]} - \alpha_{i_{3}i_{1}}^{[l_{3}]} +
\sum_{j = 3}^{n}  \alpha_{i_{j}i_{j}}^{[l_{j}]} - \alpha_{i_{j+1}i_{j}}^{[l_{j+1}]}
\mathbbm{1}_{\mathcal{E}_{1}'^{\mathrm{c}}}\big(v_{i_{j}}^{[l_{j}]}, v_{i_{j+1}}^{[l_{j+1}]} \big) \\
\label{eq:GDoF_circuit_C1_3}
d_{i_{1}}^{[2]} & \leq \alpha_{i_{1}i_{1}}^{[2]}.
\end{align}
By adding \eqref{eq:GDoF_circuit_C1_2} and \eqref{eq:GDoF_circuit_C1_3}, we obtain \eqref{eq:GDoF_circuit_C1_1}, which is therefore redundant.
If the circuit underlying the GDoF inequality in \eqref{eq:GDoF_circuit_C1_2} is also in class \ref{cond:redundant_1},
we apply the same argument above.
We do this recursively, hence showing that all circuits in class \ref{cond:redundant_1} yield redundant GDoF inequalities.

Next, after excluding all circuits in class \ref{cond:redundant_1}, suppose that we have a circuit in class \ref{cond:redundant_2} and not in class \ref{cond:redundant_1} such
that $i_{1} = i_{k}$, $k \neq 2$ and $k \neq n$ (also $k \neq 0$).
We may further assume, without loss of generality, that $l_{1} = 1$ and $l_{k} = 2$.
The corresponding GDoF inequality is given by
\begin{equation}
\label{eq:GDoF_circuit_C2_1}
\sum_{j \in \langle n \rangle} d_{i_{j}}^{[l_{j}]}
\leq
\sum_{j \in \langle n \rangle}  \alpha_{i_{j}i_{j}}^{[l_{j}]} - \alpha_{i_{j+1}i_{j}}^{[l_{j+1}]}
\end{equation}
where there is no need to employ the indicator function in \eqref{eq:GDoF_circuit_form_1} as the underlying circuit is not in \ref{cond:redundant_1}.
Now consider the circuits
$\big( v_{i_{1}}^{[l_{1}]} \rightarrow  \cdots  \rightarrow v_{i_{k}}^{[l_{k}]}  \rightarrow v_{i_{1}}^{[l_{1}]}\big)$
and $\big( v_{i_{1}}^{[l_{1}]} \rightarrow v_{i_{k+1}}^{[l_{k+1}]} \rightarrow  \cdots  \rightarrow v_{i_{n}}^{[l_{n}]}
\rightarrow v_{i_{1}}^{[l_{1}]}\big)$.
It can be easily checked that these two circuits are valid for $\mathcal{G}_{\mathrm{p}}$ and that they are not in class \ref{cond:redundant_1}.
The corresponding GDoF inequalities are given by
\begin{align}
\label{eq:GDoF_circuit_C2_2}
\sum_{j=1}^{k} d_{i_{j}}^{[l_{j}]}
& \leq \alpha_{i_{k}i_{k}}^{[l_{k}]} - \alpha_{i_{1}i_{k}}^{[l_{1}]} +
\sum_{j = 1}^{k-1}  \alpha_{i_{j}i_{j}}^{[l_{j}]} - \alpha_{i_{j+1}i_{j}}^{[l_{j+1}]} \\
\label{eq:GDoF_circuit_C2_3}
d_{i_{1}}^{[l_{1}]} + \sum_{j=k+1}^{n} d_{i_{j}}^{[l_{j}]} & \leq
\alpha_{i_{1}i_{1}}^{[l_{1}]} - \alpha_{i_{k+1}i_{1}}^{[l_{k+1}]}
+
\alpha_{i_{n}i_{n}}^{[l_{n}]} - \alpha_{i_{1}i_{n}}^{[l_{1}]}
+
\sum_{j = k+1}^{n-1}  \alpha_{i_{j}i_{j}}^{[l_{j}]} - \alpha_{i_{j+1}i_{j}}^{[l_{j+1}]}.
\end{align}
By adding the inequalities in \eqref{eq:GDoF_circuit_C2_2} and \eqref{eq:GDoF_circuit_C2_3}, while noting that $i_{k} = i_{1}$, we obtain
\begin{equation}
\label{eq:GDoF_circuit_C2_4}
d_{i_{1}}^{[l_{1}]} + \sum_{j \in \langle n \rangle} d_{i_{j}}^{[l_{j}]}
\leq
\sum_{j \in \langle n \rangle}  \alpha_{i_{j}i_{j}}^{[l_{j}]} - \alpha_{i_{j+1}i_{j}}^{[l_{j+1}]}.
\end{equation}
Comparing \eqref{eq:GDoF_circuit_C2_1} and \eqref{eq:GDoF_circuit_C2_4},
it can be seen that an extra $d_{i_{1}}^{[l_{1}]}$ is added to the left-hand-side of the latter.
Since $d_{i_{1}}^{[l_{1}]}\geq 0$, then \eqref{eq:GDoF_circuit_C2_4} implies \eqref{eq:GDoF_circuit_C2_1}.
If any of the two resulting circuits underlying the inequalities in \eqref{eq:GDoF_circuit_C2_2} and
\eqref{eq:GDoF_circuit_C2_3} is in class \ref{cond:redundant_2}, we apply the same argument above.
We do this recursively, hence showing redundancy of all circuits in class \ref{cond:redundant_2}.

Finally, suppose that we have a circuit in class \ref{cond:redundant_3} and not in \ref{cond:redundant_1} or \ref{cond:redundant_2}.
We may assume, without loss of generality, that $l_{1} = 2$ and $i_{j} \neq i_{1}$, $\forall j \in \langle 2:n \rangle$.
The corresponding GDoF inequality writes as the one in \eqref{eq:GDoF_circuit_C2_1}.
Consider the circuit given by
$\big( v_{i_{0}}^{[l_{0}]} \rightarrow v_{i_{1}}^{[l_{1}]}
\rightarrow v_{i_{1}}^{[1]} \rightarrow v_{i_{2}}^{[l_{2}]}  \rightarrow \cdots  \rightarrow v_{i_{n}}^{[l_{n}]}  \big)$,
where $v_{i_{1}}^{[1]}$ is included between $v_{i_{1}}^{[l_{1}]}$ and $ v_{i_{2}}^{[l_{2}]}$.
This is valid for $\mathcal{G}_{\mathrm{p}}$ and is not in \ref{cond:redundant_1} or \ref{cond:redundant_2}.
The corresponding GDoF inequality is given by
\begin{equation}
\label{eq:GDoF_circuit_C3_1}
d_{i_{1}}^{[1]} + \sum_{j \in \langle n \rangle} d_{i_{j}}^{[l_{j}]}
\leq
\sum_{j \in \langle n \rangle}  \alpha_{i_{j}i_{j}}^{[l_{j}]} - \alpha_{i_{j+1}i_{j}}^{[l_{j+1}]}.
\end{equation}
Comparing \eqref{eq:GDoF_circuit_C2_1} to \eqref{eq:GDoF_circuit_C3_1}, it can be seen that $d_{i_{1}}^{[1]}$ (i.e. an extra user)
is added to the left-hand-side without altering the right-hand-side.
Since $d_{i_{1}}^{[1]} \geq 0$, then \eqref{eq:GDoF_circuit_C3_1} is tighter in general.
Applying the same above argument recursively to the circuit underlying the inequality in \eqref{eq:GDoF_circuit_C3_1},
it is shown that circuits in class \ref{cond:redundant_3} are redundant.
\subsubsection{Combining Inequalities}
From single-cell circuits, we get the GDoF inequalities given by
\begin{equation}
\label{eq:GDoF_bounds_single_cell}
\sum_{s_{i} \in \langle l _{i} \rangle} d_{i}^{[s_{i}]} \leq \alpha_{ii}^{[l_{i}]}, \ \forall i \in \langle K \rangle, l_{i} \in \{1,2\}
\end{equation}
On the other hand, we only need to consider multi-cell circuits that do not traverse $u$ and do
not belong to any of the classes  \ref{cond:redundant_1}--\ref{cond:redundant_3}.
From such circuits, we get the inequalities
\begin{align}
\nonumber
\sum_{j \in \langle m \rangle}\sum_{s_{i_{j}} \in \langle l _{i_{j}} \rangle} d_{i_{j}}^{[s_{i_{j}}]} & \leq
\sum_{j \in \langle m \rangle}\alpha_{i_{j}i_{j}}^{[l_{i_{j}}]} - \alpha_{i_{j+1}i_{j}}^{[l_{i_{j+1}}]}
\stackrel{\text{(a)}}{=}
\sum_{j \in \langle m \rangle}\alpha_{i_{j}i_{j}}^{[l_{i_{j}}]} - \alpha_{i_{j}i_{j-1}}^{[l_{i_{j}}]}, \\
\label{eq:GDoF_bounds_multi_cell}
\forall l_{i_{j}} \in \{1,2\}, \ & (i_{1},\ldots,i_{m}) \in \Sigma\big(\langle K \rangle\big), m \in \langle 2:K \rangle
\end{align}
where (a) follows by rearranging the terms.
Combining \eqref{eq:GDoF_bounds_single_cell} and \eqref{eq:GDoF_bounds_multi_cell} with the non-negativity constraint on
$d_{i}^{[l_{i}]}$, $\forall (l_{i},i) \in \mathcal{K}$, leads directly to the characterization in Theorem \ref{theorem:TIN_region}.
\section{TIN Optimality}
\label{sec:TIN Optimality}
\subsection{Outer Bound}
The TIN-optimality result in Theorem \ref{theorem:TIN_optimality} follows directly from the following outer bound.
\begin{theorem}
\label{theorem:outer_bound}
For the IMAC with input-output relationship in \eqref{eq:system model 2}, if the TIN-optimality conditions in \eqref{eq:TIN_condition_1}
and \eqref{eq:TIN_condition_2} hold, then the capacity region $\mathcal{C}$ is included in the set of rate tuples satisfying
\begin{align}
\label{eq:capacity_outer_1}
\sum_{s_{i}\in \langle l_{i} \rangle}R_{i}^{[s_{i}]} & \leq \log\left(1 + l_{i}P^{\alpha_{ii}^{[l_{i}]}} \right),
\ (l_{i},i) \in \mathcal{K}\\
\nonumber
\sum_{j \in \langle m \rangle } \sum_{s_{i_{j}} \in \langle l_{i_{j}} \rangle} R_{i_{j}}^{[s_{i_{j}}]} & \leq
m + \sum_{j \in \langle m \rangle } \log\left( 1 +
(l_{i_{j}+1} + l_{i_{j}}) P^{\alpha_{i_{j}i_{j}}^{[l_{i_{j}}]}-\alpha_{i_{j}i_{j-1}}^{[l_{i_{j}}]}}  \right), \\
\label{eq:capacity_outer_2}
\forall l_{i_{j}} \in \{1,2\}, \ & (i_{1},\ldots,i_{m}) \in \Sigma\big(\langle K \rangle\big), m \in \langle 2:K \rangle.
\end{align}
\end{theorem}
\begin{proof}
For each cell $i$, \eqref{eq:capacity_outer_1} is a cut-set upper bound and
follows from the MAC capacity region \cite{Cover2012} and \eqref{eq:strength_order}. Hence, we focus on the cyclic bounds in \eqref{eq:capacity_outer_2}.

Cells and users participating in a given cyclic bound are identified by the two sequences
$(i_{1},\ldots,i_{m}) \in \Sigma\big(\langle K \rangle\big)$ and $(l_{i_{1}},\ldots,l_{i_{m}}) \in \{1,2\}^{m}$.
Given such sequences, each participating cell $i_{j}$ is in one of the three following subsets:
$\mathcal{S}_{1} \triangleq \{i_{j}: l_{i_{j}} = 1 \}$,
$\mathcal{S}_{2} \triangleq \{i_{j}: l_{i_{j}} = 2,\; \alpha_{i_{j}i_{j-1}}^{[1]} \leq \alpha_{i_{j}i_{j-1}}^{[2]} \}$
and
$\mathcal{S}_{3} \triangleq \{i_{j}: l_{i_{j}} = 2,\; \alpha_{i_{j}i_{j-1}}^{[1]} >\alpha_{i_{j}i_{j-1}}^{[2]} \}$.
Next, we go through the following steps
\begin{itemize}
\item Eliminate all non-participating transmitters $(l_{i},i) \in \mathcal{K} \setminus \left\{(s_{i_{j}},i_{j}): s_{i_{j}} \in
\langle l_{i_{j}} \rangle,j \in \langle m \rangle \right\}$, all non-participating receivers
$i \in \langle K \rangle \setminus \{i_{1},\ldots,i_{m} \}$ and the corresponding messages.
\item Eliminate all interfering links except for links from
Tx-$(l_{i_{j}},i_{j})$ to Rx-$i_{j-1}$, $\forall j \in \langle m \rangle$,
and from Tx-$(1,i_{j})$ to Rx-$i_{j-1}$, $\forall i_{j} \in \mathcal{S}_{3}$.
\end{itemize}
We end up with a partially connected cyclic IMAC with input-output relationship
\begin{equation}
\label{eq:system model_cyclic}
Y_{i_{j}}(t) = \sum_{s_{i_{j}} \in \langle l_{i_{j}} \rangle} h_{i_{j}i_{j}}^{[s_{i_{j}}]} \tilde{X}_{i_{j}}^{[s_{i_{j}}]}(t) + U_{i_{j+1}}(t)
\end{equation}
where the interference plus noise $U_{i_{j}}(t)$, caused by cell $i_{j}$ to cell $i_{j-1}$, is given by
\begin{equation}
\label{eq:U_ij_t}
U_{i_{j}}(t) =
\begin{cases}
h_{i_{j}i_{j-1}}^{[l_{i_{j}}]} \tilde{X}_{i_{j}}^{[l_{i_{j}}]}(t) + Z_{i_{j-1}}(t), \ i_{j} \in \mathcal{S}_{1} \cup \mathcal{S}_{2} \\
h_{i_{j}i_{j-1}}^{[1]} \tilde{X}_{i_{j}}^{[1]}(t) + h_{i_{j}i_{j-1}}^{[2]} \tilde{X}_{i_{j}}^{[2]}(t) + Z_{i_{j-1}}(t),
\ i_{j}\in \mathcal{S}_{3}.
\end{cases}
\end{equation}
Since none of the above steps hurts the rates of the remaining messages, the channel in \eqref{eq:system model_cyclic}
is used for the outer bound.
From \eqref{eq:system model_cyclic} onwards, we revert back to the original channel notation for notational convenience, while maintaining
$|h_{ki}^{[l_{k}]} |^{2} P_{k}^{[l_{k}]} \geq 1$, $\forall (l_{k},k) \in \mathcal{K}, i\in \langle K \rangle$.

Next, we define the side information signal $S_{i_{j}}(t)$, $\forall j \in \langle m \rangle$, as
\begin{equation}
\nonumber
S_{i_{j}}(t) =
\begin{cases}
U_{i_{j}}(t), \ i_{j} \in \mathcal{S}_{1} \cup \mathcal{S}_{2} \\
\frac{h_{i_{j}i_{j-1}}^{[2]}}{h_{i_{j}i_{j}}^{[2]}}\left(h_{i_{j}i_{j}}^{[1]} \tilde{X}_{i_{j}}^{[1]}(t) +
h_{i_{j}i_{j}}^{[2]} \tilde{X}_{i_{j}}^{[2]}(t) \right) + Z_{i_{j-1}}(t),
\ i_{j}\in \mathcal{S}_{3}
\end{cases}
\end{equation}
and we provide $S_{i_{j}}^{n}$ for Rx-$i_{j}$ through a genie.
From Fano's inequality, we have
\begin{align}
\nonumber
n \sum_{s_{i_{j}} \in \langle l_{i_{j}} \rangle} R_{i_{j}}^{[s_{i_{j}}]} -  n \epsilon & \leq I
\big(W_{i_{j}}^{[1:l_{i_{j}}]} ; Y_{i_{j}}^{n}, S_{i_{j}}^{n}   \big) \\
\nonumber
&  = I\big(W_{i_{j}}^{[1:l_{i_{j}}]} ; S_{i_{j}}^{n}   \big) +
I\big(W_{i_{j}}^{[1:l_{i_{j}}]} ; Y_{i_{j}}^{n} | S_{i_{j}}^{n}   \big) \\
\nonumber
& = h\big(S_{i_{j}}^{n}\big) - h\big(S_{i_{j}}^{n}  | W_{i_{j}}^{[1:l_{i_{j}}]}\big) +
h\big(Y_{i_{j}}^{n}  | S_{i_{j}}^{n}\big) - h\big( Y_{i_{j}}^{n} |  S_{i_{j}}^{n} , W_{i_{j}}^{[1:l_{i_{j}}]} \big) \\
\label{eq:Fano_and_after}
& = h\big(S_{i_{j}}^{n}\big) - h\big(Z_{i_{j-1}}^{n} \big) +
h\big(Y_{i_{j}}^{n}  | S_{i_{j}}^{n}\big) - h\big( U_{i_{j+1}}^{n} \big)
\end{align}
where $W_{i_{j}}^{[1:l_{i_{j}}]} \triangleq W_{i_{j}}^{[1]},\ldots,W_{i_{j}}^{[l_{i_{j}}]}$.
Proceeding from \eqref{eq:Fano_and_after}, we have
\begin{equation}
\label{eq:sum_rate_bound}
n \sum_{j \in \langle m \rangle}\sum_{s_{i_{j}} \in \langle l_{i_{j}} \rangle} R_{i_{j}}^{[s_{i_{j}}]} - nm \epsilon   \leq
 \sum_{j \in \langle m \rangle} \left[ h\big(S_{i_{j}}^{n}\big)  - h\big(U_{i_{j}}^{n}\big) +
  h\big(Y_{i_{j}}^{n}  | S_{i_{j}}^{n}\big) - h\big(Z_{i_{j}}^{n} \big) \right].
\end{equation}
Considering the first difference of entropies in \eqref{eq:sum_rate_bound} for a given $j \in \langle m \rangle$,
it is clear that this  is equal to $0$ if $i_{j} \in \mathcal{S}_{1} \cup \mathcal{S}_{2} $.
Hence, we focus on $i_{j} \in \mathcal{S}_{3}$.
For this case, and from \eqref{eq:TIN_condition_2}, we have
\begin{equation}
\label{eq:condition_TIN_S3}
\alpha_{i_{j}i_{j}}^{[2]} - 2\alpha_{i_{j}i_{j-1}}^{[2]} \geq  \alpha_{i_{j}i_{j}}^{[1]} - \alpha_{i_{j}i_{j-1}}^{[1]}
\Leftrightarrow
\frac{P_{i_{j}}^{[1]}\big|h_{i_{j}i_{j-1}}^{[1]} \big|^{2}}{P_{i_{j}}^{[2]}\big|h_{i_{j}i_{j-1}}^{[2]} \big|^{2}} \geq
P_{i_{j}}^{[1]}\big|h_{i_{j}i_{j}}^{[1]} \big|^{2}\frac{\big|h_{i_{j}i_{j-1}}^{[2]} \big|^{2}}{\big|h_{i_{j}i_{j}}^{[2]} \big|^{2}}.
\end{equation}
The condition in \eqref{eq:condition_TIN_S3} allows us to apply \cite[Lemma 8]{Gherekhloo2016}, from which we obtain
\begin{equation}
\label{eq:difference_h_S_U}
 h\big(S_{i_{j}}^{n}\big)  - h\big(U_{i_{j}}^{n}\big)  \leq n.
\end{equation}
Now we turn our attention to the second difference of entropies in \eqref{eq:sum_rate_bound}. We have
\begin{align}
\nonumber
h\big(Y_{i_{j}}^{n}  | S_{i_{j}}^{n}\big) - h\big(Z_{i_{j}}^{n} \big) & \leq
\sum_{t \in \langle n \rangle} \big[ h\big(Y_{i_{j}}(t)  | S_{i_{j}}(t)\big) - h\big(Z_{i_{j}}(t) \big)  \big] \\
\label{eq:conditional_entropy_Gaussian}
& \leq \sum_{t \in \langle n \rangle} \big[ h\big(Y_{i_{j}}^{\mathrm{G}}(t)  | S_{i_{j}}^{\mathrm{G}}(t)\big) - h\big(Z_{i_{j}}(t) \big) \big] \\
\label{eq:conditional_entropy_Gaussian_2}
& \leq n \log \big( \sigma^{2}_{Y_{i_{j}}^{\mathrm{G}}  | S_{i_{j}}^{\mathrm{G}}} \big)
\end{align}
where $\mathrm{G}$ indicates that the corresponding inputs are i.i.d Gaussian, i.e.
$\tilde{X}_{i}^{[l_{i}]} \sim \mathcal{N}_{\mathbb{C}}\big(0,P_{i}^{[l_{i}]}\big)$,
and
\begin{equation}
\sigma^{2}_{Y_{i_{j}}^{\mathrm{G}}  | S_{i_{j}}^{\mathrm{G}}} \triangleq
\E \big[ |Y_{i_{j}}^{\mathrm{G}}|^{2} \big] - \E \big[ Y_{i_{j}}^{\mathrm{G}}S_{i_{j}}^{\mathrm{G}\ast} \big]
\E \big[ S_{i_{j}}^{\mathrm{G}}Y_{i_{j}}^{\mathrm{G}\ast} \big]
\big(\E \big[ |S_{i_{j}}^{\mathrm{G}}|^{2} \big]\big)^{-1}.
\end{equation}
Note that we omit the time index $t$ from \eqref{eq:conditional_entropy_Gaussian_2} onwards for brevity.
The inequality in \eqref{eq:conditional_entropy_Gaussian} follows because Gaussian distribution maximizes the
conditional differential entropy for a given covariance constraint.
Next, we calculate $\sigma^{2}_{Y_{i_{j}}^{\mathrm{G}}  | S_{i_{j}}^{\mathrm{G}}}$ as
\begin{equation}
\label{eq:sigma_conditional}
\sigma^{2}_{Y_{i_{j}}^{\mathrm{G}}  | S_{i_{j}}^{\mathrm{G}}}
=
\begin{cases}
\E\big[|U_{i_{j+1}}^{\mathrm{G}}|^{2}\big] + \frac{P_{i_{j}}^{[1]}\big|h_{i_{j}i_{j}}^{[1]}\big|^{2} }
{1 +P_{i_{j}}^{[1]}\big|h_{i_{j}i_{j-1}}^{[1]}\big|^{2} }, \ i_{j}\in \mathcal{S}_{1} \\
\E\big[|U_{i_{j+1}}^{\mathrm{G}}|^{2}\big] + P_{i_{j}}^{[1]}\big|h_{i_{j}i_{j}}^{[1]}\big|^{2}+
 \frac{P_{i_{j}}^{[2]}\big|h_{i_{j}i_{j}}^{[2]}\big|^{2} }
{1 +P_{i_{j}}^{[2]}\big|h_{i_{j}i_{j-1}}^{[2]}\big|^{2} }, \
i_{j}\in \mathcal{S}_{2} \\
\E\big[|U_{i_{j+1}}^{\mathrm{G}}|^{2}\big] +
 \frac{P_{i_{j}}^{[1]}\big|h_{i_{j}i_{j}}^{[1]}\big|^{2} + P_{i_{j}}^{[2]}\big|h_{i_{j}i_{j}}^{[2]}\big|^{2} }
{1 + \frac{ \big| h_{i_{j}i_{j-1}}^{[2]} \big|^{2} }{\big| h_{i_{j}i_{j}}^{[2]} \big|^{2}}
\big( P_{i_{j}}^{[1]}\big|h_{i_{j}i_{j}}^{[1]}\big|^{2} + P_{i_{j}}^{[2]}\big|h_{i_{j}i_{j}}^{[2]}\big|^{2} \big) }, \
i_{j}\in \mathcal{S}_{3}.
\end{cases}
\end{equation}
The expressions for the three cases in \eqref{eq:sigma_conditional} are bounded above as
\begin{equation}
\label{eq:sigma_conditional_UB_0}
\sigma^{2}_{Y_{i_{j}}^{\mathrm{G}}  | S_{i_{j}}^{\mathrm{G}}}
\leq
\begin{cases}
1 + P_{i_{j+1}}^{[1]}\big|h_{i_{j+1}i_{j}}^{[1]}\big|^{2} + \frac{P_{i_{j}}^{[1]}\big|h_{i_{j}i_{j}}^{[1]}\big|^{2} }
{P_{i_{j}}^{[1]}\big|h_{i_{j}i_{j-1}}^{[1]}\big|^{2} }, \ i_{j}\in \mathcal{S}_{1} \\
1 + P_{i_{j+1}}^{[2]}\big|h_{i_{j+1}i_{j}}^{[2]}\big|^{2} + P_{i_{j}}^{[1]}\big|h_{i_{j}i_{j}}^{[1]}\big|^{2}+
 \frac{P_{i_{j}}^{[2]}\big|h_{i_{j}i_{j}}^{[2]}\big|^{2} }
{P_{i_{j}}^{[2]}\big|h_{i_{j}i_{j-1}}^{[2]}\big|^{2} }, \
i_{j}\in \mathcal{S}_{2} \\
1 + P_{i_{j+1}}^{[1]}\big|h_{i_{j+1}i_{j}}^{[1]}\big|^{2} +
P_{i_{j+1}}^{[2]}\big|h_{i_{j+1}i_{j}}^{[2]}\big|^{2}  +
 \frac{2 P_{i_{j}}^{[2]}\big|h_{i_{j}i_{j}}^{[2]}\big|^{2} }
{ P_{i_{j}}^{[2]} \big| h_{i_{j}i_{j-1}}^{[2]} \big|^{2}  }, \
i_{j}\in \mathcal{S}_{3}
\end{cases}
\end{equation}
where we have employed \eqref{eq:U_ij_t} and the order of strengths \eqref{eq:strength_order}.
Converting to the notation of \eqref{eq:system model 2}
and employing the TIN conditions in \eqref{eq:TIN_condition_1} and \eqref{eq:TIN_condition_2},
we obtain a further upper bound for \eqref{eq:sigma_conditional_UB_0} as
\begin{equation}
\label{eq:sigma_conditional_UB}
\sigma^{2}_{Y_{i_{j}}^{\mathrm{G}}  | S_{i_{j}}^{\mathrm{G}}}
\leq
1 + (l_{i_{j+1}} + l_{i_{j}})P^{\alpha_{i_{j}i_{j}}^{[l_{j}]} -  \alpha_{i_{j}i_{j-1}}^{[l_{j}]}}.
\end{equation}
By combining \eqref{eq:sum_rate_bound} with \eqref{eq:difference_h_S_U}, \eqref{eq:conditional_entropy_Gaussian_2} and \eqref{eq:sigma_conditional_UB},
we obtain the bound in \eqref{eq:capacity_outer_2}.
\end{proof}
\subsection{Constant Gap to Capacity Region}
Utilizing Theorem \ref{theorem:TIN_region}, Theorem \ref{theorem:TIN_optimality} and Theorem \ref{theorem:outer_bound}, it can be
shown that the proposed TIN scheme can achieve the whole capacity region to within a constant gap at any finite SNR.
\begin{theorem}
For the IMAC with input-output relationship in \eqref{eq:system model 2}, if the TIN-optimality conditions in \eqref{eq:TIN_condition_1}
and \eqref{eq:TIN_condition_2} hold, then the rate region achieved through
the TIN scheme with decoding order $\bm{\pi} = \bm{\mathrm{id}}$, as described in Section \ref{subsec:TIN_scheme}, is within $2 + \log(5K)$ bits of the capacity region $\mathcal{C}$.
\end{theorem}
\begin{proof}
The above result is shown by following the same steps used to prove \cite[Theorem 4]{Geng2015}.
First, we obtain an outer bound which is within a constant gap from the one in Theorem \ref{theorem:outer_bound}.
For the bound in \eqref{eq:capacity_outer_1}, since  $P > 1$ and $l_{i} \leq 2$, we have
\begin{align}
\nonumber
\sum_{s_{i}\in \langle l_{i} \rangle}R_{i}^{[s_{i}]} & \leq \log\big(1 + l_{i}P^{\alpha_{ii}^{[l_{i}]}} \big)\\
\nonumber
& \leq \log \big(3 P^{\alpha_{ii}^{[l_{i}]}} \big) = \log(3 ) +\log\big( P^{\alpha_{ii}^{[l_{i}]}} \big).
\end{align}
On the other hand, for the bound in \eqref{eq:capacity_outer_2}, it follows that
\begin{align}
\nonumber
\sum_{j \in \langle m \rangle } \sum_{s_{i_{j}} \in \langle l_{i_{j}} \rangle} R_{i_{j}}^{[s_{i_{j}}]} & \leq
\sum_{j \in \langle m \rangle } \Big[ 1 +  \log\Big( 1 +
(l_{i_{j}+1} + l_{i_{j}}) P^{\alpha_{i_{j}i_{j}}^{[l_{i_{j}}]}-\alpha_{i_{j}i_{j-1}}^{[l_{i_{j}}]}}  \Big)  \Big]\\
\nonumber
& \leq \sum_{j \in \langle m \rangle } \Big[ 1 + \log(5) +  \log\Big( P^{\alpha_{i_{j}i_{j}}^{[l_{i_{j}}]}-\alpha_{i_{j}i_{j-1}}^{[l_{i_{j}}]}}  \Big)  \Big].
\end{align}
From the above, we see that $\mathcal{C}$ is contained in the region given by all rate tuples $\mathbf{R} \in \mathbb{R}_{+}$ such that
\begin{align}
\label{eq:capacity_outer_gap_1}
\sum_{s_{i}\in \langle l_{i} \rangle}R_{i}^{[s_{i}]} & \leq \alpha_{ii}^{[l_{i}]}\log(P ) +  \log(10),
\ (l_{i},i) \in \mathcal{K}\\
\nonumber
\sum_{j \in \langle m \rangle } \sum_{s_{i_{j}} \in \langle l_{i_{j}} \rangle} R_{i_{j}}^{[s_{i_{j}}]} & \leq
\sum_{j \in \langle m \rangle } \Big[ (\alpha_{i_{j}i_{j}}^{[l_{i_{j}}]}-\alpha_{i_{j}i_{j-1}}^{[l_{i_{j}}]})  \log ( P ) +  \log(10) \Big],
\\
\label{eq:capacity_outer_gap_2}
\forall l_{i_{j}} \in \{1,2\}, \ & (i_{1},\ldots,i_{m}) \in \Sigma\big(\langle K \rangle\big), m \in \langle 2:K \rangle.
\end{align}

Next, we derive an achievable rate region.
We fix the decoding order to $\bm{\pi} = \bm{\mathrm{id}}$.
From \eqref{eq:rate per user_1} and  \eqref{eq:rate per user_2},
we know that for any feasible power allocation $\mathbf{r}$, the rate tuple $\bar{\mathbf{R}} \in \mathbb{R}_{+}$
that satisfies
\begin{align}
\label{eq:rate_gap_1}
\bar{R}_{k}^{[1]} & = \log \Biggl( 1+ \frac{ P^{ r_{k}^{[1]} +  \alpha_{kk}^{[1]} }  }
{1  +
\sum_{j\neq k} \big[ P^{ r_{j}^{[1]} + \alpha_{jk}^{[1]} }+ P^{ r_{j}^{[2]} + \alpha_{jk}^{[2]} }  \big]  } \Biggr)  \\
\label{eq:rate_gap_2}
\bar{R}_{k}^{[1]}  + \bar{R}_{k}^{[2]} & =\log \Biggl( 1+ \frac{ P^{ r_{k}^{[1]} + \alpha_{kk}^{[1]} } +  P^{ r_{k}^{[2]} + \alpha_{kk}^{[2]}  }  }
{1    +
\sum_{j\neq k} \big[ P^{ r_{j}^{[1]} + \alpha_{jk}^{[1]} }+ P^{ r_{j}^{[2]} + \alpha_{jk}^{[2]} }  \big]  } \Biggr).
\end{align}
for all $k \in \langle K \rangle$, is achievable.
The region given by all such rate tuples, corresponding to all feasible $\mathbf{r}$, is hence achievable.
Next, we characterize this rate region to within a constant gap.
From the proof of Theorem \ref{theorem:TIN_region} (see Section \ref{sec:TIN_region}) and Theorem \ref{theorem:TIN_optimality}, we know that when the conditions in \eqref{eq:TIN_condition_1} and \eqref{eq:TIN_condition_2} hold,
$\mathbf{d} \in \mathcal{P}_{\bm{\mathrm{id}}}$ if and only if there exists a power allocation $\mathbf{r}$
such that
\begin{align}
\label{eq:rate_gap_power_1}
d_{k}^{[1]}  & = r_{k}^{[1]} + \alpha_{kk}^{[1]}
 -  \max \Bigl\{ 0  , \max_{j \neq k} \bigl\{  \max_{l_{j}} \{ r_{j}^{[l_{j}]} + \alpha_{jk}^{[l_{j}]} \}  \bigr\} \Bigr\},
 \ k \in \langle K \rangle   \\
\label{eq:rate_gap_power_1_2}
d_{k}^{[2]} & =  r_{k}^{[2]}  +   \alpha_{kk}^{[2]}
  -  \max \Bigl\{ 0  , r_{k}^{[1]}  +   \alpha_{kk}^{[1]} ,\max_{j \neq k} \bigl\{
  \max_{l_{j}} \{ r_{j}^{[l_{j}]}  +  \alpha_{jk}^{[l_{j}]} \}  \bigr\} \Bigr\},  \ k \in \langle K \rangle  \\
\label{eq:rate_gap_power_3}
r_{k}^{[l_{k}]} & \leq 0 , \   (l_{k},k) \in \mathcal{K}
\end{align}
are satisfied\footnote{Note that while the conditions for feasible power allocation in \eqref{eq:GDoF_polyhedral_1} and \eqref{eq:GDoF_polyhedral_2} (and hence \eqref{eq:polyhedral_region_21}--\eqref{eq:polyhedral_region_25}) are given in terms of inequalities, equality in \eqref{eq:rate_gap_power_1} and \eqref{eq:rate_gap_power_2} can be shown using the fixed point theorem as in \cite[Appendix B]{Geng2015}.}.
By adding \eqref{eq:rate_gap_power_1} and \eqref{eq:rate_gap_power_1_2}, we obtain
\begin{align}
\nonumber
d_{k}^{[1]} + d_{k}^{[2]} & =  r_{k}^{[1]} + \alpha_{kk}^{[1]} + r_{k}^{[2]}  +   \alpha_{kk}^{[2]}
 -  \max \Bigl\{ 0  , \max_{j \neq k} \bigl\{  \max_{l_{j}} \{ r_{j}^{[l_{j}]} + \alpha_{jk}^{[l_{j}]} \}  \bigr\} \Bigr\} \\
\nonumber
& \quad \quad -  \max \Bigl\{ 0  , r_{k}^{[1]}  +   \alpha_{kk}^{[1]} ,\max_{j \neq k} \bigl\{
  \max_{l_{j}} \{ r_{j}^{[l_{j}]}  +  \alpha_{jk}^{[l_{j}]} \}  \bigr\} \Bigr\}  \\
\nonumber
  & \leq  r_{k}^{[2]}  +   \alpha_{kk}^{[2]}
 -  \max \Bigl\{ 0  , \max_{j \neq k} \bigl\{  \max_{l_{j}} \{ r_{j}^{[l_{j}]} + \alpha_{jk}^{[l_{j}]} \}  \bigr\} \Bigr\} \\
\label{eq:rate_gap_power_2}
& \leq \max \bigl\{r_{k}^{[1]}  +   \alpha_{kk}^{[1]}, r_{k}^{[2]}  +   \alpha_{kk}^{[2]} \bigr\}
  -  \max \Bigl\{ 0  ,\max_{j \neq k} \bigl\{
  \max_{l_{j}} \{ r_{j}^{[l_{j}]}  +  \alpha_{jk}^{[l_{j}]} \}  \bigr\} \Bigr\}.
\end{align}
We employ the above to characterize the achievable rate region.
In particular, from \eqref{eq:rate_gap_power_1}, the achievable rate in \eqref{eq:rate_gap_1} is bounded below as
\begin{align}
\nonumber
\bar{R}_{k}^{[1]} & \geq \log \Biggl(  \frac{ P^{ r_{k}^{[1]} +  \alpha_{kk}^{[1]} }  }
{P^{0}  +
\sum_{j\neq k} \big[ P^{ r_{j}^{[1]} + \alpha_{jk}^{[1]} }+ P^{ r_{j}^{[2]} + \alpha_{jk}^{[2]} }  \big]  } \Biggr)  \\
\nonumber
& \geq  \log \Biggl(  \frac{ P^{ r_{k}^{[1]} +  \alpha_{kk}^{[1]} }  }
{ [1+ 2(K-1)] P^{ \max \bigl\{ 0  , \max_{j \neq k} \{  \max_{l_{j}} \{ r_{j}^{[l_{j}]} + \alpha_{jk}^{[l_{j}]} \}  \} \bigr\} } } \Biggr) \\
\nonumber
& \geq  \log \Biggl(  \frac{ P^{ r_{k}^{[1]} +  \alpha_{kk}^{[1]} }  }
{ [1+ 2(K-1)] P^{ r_{k}^{[1]} + \alpha_{kk}^{[1]} - d_{k}^{[1]} } } \Biggr) \\
\nonumber
& \geq  d_{k}^{[1]} \log(P) - \log(2K).
\end{align}
Similarly, from \eqref{eq:rate_gap_power_2}, the sum rate in \eqref{eq:rate_gap_2} is bounded below as
\begin{align}
\nonumber
\bar{R}_{k}^{[1]}  + \bar{R}_{k}^{[2]} & \geq \log \Biggl( \frac{ P^{ r_{k}^{[1]} + \alpha_{kk}^{[1]} } +  P^{ r_{k}^{[2]} + \alpha_{kk}^{[2]}  }  }
{P^{0}    +
\sum_{j\neq k} \big[ P^{ r_{j}^{[1]} + \alpha_{jk}^{[1]} }+ P^{ r_{j}^{[2]} + \alpha_{jk}^{[2]} }  \big]  } \Biggr) \\
\nonumber
& \geq  \log \Biggl(  \frac{ P^{ \max \bigl\{r_{k}^{[1]}  +   \alpha_{kk}^{[1]}, r_{k}^{[2]}  +   \alpha_{kk}^{[2]} \bigr\} }  }
{ [1+ 2(K-1)] P^{ \max \bigl\{ 0  , \max_{j \neq k} \{  \max_{l_{j}} \{ r_{j}^{[l_{j}]} + \alpha_{jk}^{[l_{j}]} \}  \} \bigr\} } } \Biggr) \\
\nonumber
& \geq  \log \Biggl(  \frac{ P^{ \max \bigl\{r_{k}^{[1]}  +   \alpha_{kk}^{[1]}, r_{k}^{[2]}  +   \alpha_{kk}^{[2]} \bigr\} }  }
{ [1+ 2(K-1)] P^{ \max \bigl\{r_{k}^{[1]}  +   \alpha_{kk}^{[1]}, r_{k}^{[2]}  +   \alpha_{kk}^{[2]} \bigr\} -
\big(d_{k}^{[1]} + d_{k}^{[2]} \big) } } \Biggr) \\
\nonumber
& \geq  \big(d_{k}^{[1]} + d_{k}^{[2]} \big) \log(P) - \log(2K).
\end{align}
From the above, and since $\mathbf{d} \in \mathcal{P}_{\bm{\mathrm{id}}}$,
the achievable rate region, as specified through \eqref{eq:rate_gap_1} and \eqref{eq:rate_gap_2}, contains the region given by
all rate tuples $\bar{\mathbf{R}} \in \mathbb{R}_{+}$ that satisfy
\begin{align}
\label{eq:capacity_inner_gap_1}
\sum_{s_{i}\in \langle l_{i} \rangle}\bar{R}_{i}^{[s_{i}]} & \leq
\max\left\{ 0, \alpha_{ii}^{[l_{i}]}\log(P ) - \log(2K) \right\},
\ (l_{i},i) \in \mathcal{K}\\
\nonumber
\sum_{j \in \langle m \rangle } \sum_{s_{i_{j}} \in \langle l_{i_{j}} \rangle} \bar{R}_{i_{j}}^{[s_{i_{j}}]} & \leq
\max\left\{ 0 ,
\sum_{j \in \langle m \rangle } \Big[ (\alpha_{i_{j}i_{j}}^{[l_{i_{j}}]}-\alpha_{i_{j}i_{j-1}}^{[l_{i_{j}}]})  \log ( P ) - \log(2K) \Big]
\right\}, \\
\label{eq:capacity_inner_gap_2}
\forall l_{i_{j}} \in \{1,2\}, \ & (i_{1},\ldots,i_{m}) \in \Sigma\big(\langle K \rangle\big), m \in \langle 2:K \rangle.
\end{align}
At this point, it can be easily shown that each of the rate constraints
in  \eqref{eq:capacity_inner_gap_1} and \eqref{eq:capacity_inner_gap_2} is within at most
$\log(20K)$ bits (per dimension) of its corresponding outer bound in \eqref{eq:capacity_outer_gap_1} and \eqref{eq:capacity_outer_gap_2} (e.g. see the proof of \cite[Theorem 4]{Geng2015}).
This completes the proof of the theorem.
\end{proof}
\section{Concluding Remarks}
In this paper, we considered the TIN optimality problem for the Gaussian IMAC.
We derived a TIN-achievable GDoF region through a novel application of the potential theorem approach in \cite{Geng2015,Geng2016}.
Moveover, we proved the optimality of this GDoF region for a non-trivial regime of parameters by building
upon the genie-aided converse arguments in \cite{Etkin2008}, \cite{Geng2015} and \cite{Gherekhloo2016}.
An interesting extension  is to consider the more general scenario where each
MAC consists of an arbitrary number of users.
\bibliographystyle{IEEEtran}
\bibliography{References}
\end{document}